\documentclass[a4paper,UKenglish]{lipics-v2019}
\usepackage{microtype}%if unwanted, comment out or use option "draft"

\usepackage[utf8]{inputenc}
\usepackage{amsmath}
\usepackage{amsfonts}
\usepackage{amsthm}
\usepackage{bm}
\usepackage{algorithm}
\usepackage{algorithmic}
\usepackage{comment}
\usepackage{color}
\usepackage{graphics}

\renewcommand\proofname{\bf Proof}

\newtheorem{prp}{Proposition}
\newtheorem{prb}{Problem}

\newcommand{\Cost}{\mathrm{Cost}}
\newcommand{\DP}{\mathrm{DP}}

\title{$r$-Gather Clustering and $r$-Gathering on Spider: FPT Algorithms and Hardness}

\author{Soh Kumabe}{The University of Tokyo \\ RIKEN AIP}{soh\_kumabe@mist.i.u-tokyo.ac.jp}{}{}%mandatory, please use full name; only 1 author per \author macro; first two parameters are mandatory, other parameters can be empty.
\author{Takanori Maehara}{RIKEN AIP}{takanori.maehara@riken.jp}{}{}

\authorrunning{S.\,Kumabe and T.\,Maehara}%mandatory. First: Use abbreviated first/middle names. Second (only in severe cases): Use first author plus 'et al.'

\Copyright{Soh Kumabe and Takanori Maehara}%mandatory, please use full first names. LIPIcs license is "CC-BY";  http://creativecommons.org/licenses/by/3.0/


%mandatory

\category{}%optional, e.g. invited paper

\relatedversion{}%optional, e.g. full version hosted on arXiv, HAL, or other respository/website

\supplement{}%optional, e.g. related research data, source code, ... hosted on a repository like zenodo, figshare, GitHub, ...

\funding{}%optional, to capture a funding statement, which applies to all authors. Please enter author specific funding statements as fifth argument of the \author macro.

\acknowledgements{}%optional

\EventEditors{John Q. Open and Joan R. Access}
\EventNoEds{2}
\EventLongTitle{42nd Conference on Very Important Topics (CVIT 2016)}
\EventShortTitle{CVIT 2016}
\EventAcronym{CVIT}
\EventYear{2016}
\EventDate{December 24--27, 2016}
\EventLocation{Little Whinging, United Kingdom}
\EventLogo{}
\SeriesVolume{42}
\ArticleNo{23}
\nolinenumbers %uncomment to disable line numbering
\hideLIPIcs  %uncomment to remove references to LIPIcs series (logo, DOI, ...), e.g. when preparing a pre-final version to be uploaded to arXiv or another public repository
\begin{document}

\maketitle

\begin{abstract}
We consider min-max $r$-gather clustering problem and min-max $r$-gathering problem.
In the min-max $r$-gather clustering problem, we are given a set of users and divide them into clusters with size at least $r$; the goal is to minimize the maximum diameter of clusters.
In the min-max $r$-gathering problem, we are additionally given a set of facilities and assign each cluster to a facility; the goal is to minimize the maximum distance between the users and the assigned facility.
%In this study, we consider the case that all users and facilities are located on the spider-shaped metric.
In this study, we consider the case that the users and facilities are located on a ``spider'' and propose the first fixed-parameter tractable (FPT) algorithms for both problems, which are parametrized by only the number of legs.
Furthermore, we prove that these problems are NP-hard when the number of legs is arbitrarily large.
\end{abstract}

\section{Introduction}

\subsection{Background and Motivation}

We consider min-max $r$-gather clustering problem and min-max $r$-gathering problem.
In the \emph{min-max $r$-gather clustering problem},
we are given a metric space $\mathcal{M}$ and $n$ users located on $\mathcal{M}$. 
The goal of the problem is to find a partition of the users such that each cluster contains at least $r$ users, and the maximum diameter of the clusters is minimized\footnote{Several variations of $r$-gather clustering problem are known.
At first appearance, it is the special case of $r$-gathering problem, where the set of locations of the users are the subset of the set of locations of the facilities~\cite{aggarwal2010achieving}.
In recent researches on the restricted metrics, the same definition as this paper is often used.
If the metric is a subclass of the tree, that corresponds to the situation, where we can open facilities at any points on the metric.
As the clustering problem, this is a more natural way to define the cost, since we do not need the set of facilities as an input.
Thus we use this definition.
}~\cite{akagi2015r,nakano2018simple,ahmed2019r}. 
In the \emph{min-max $r$-gathering problem}, we are additionally given $m$ facilities on $\mathcal{M}$. 
The goal of the problem is to open a subset of facilities and assign each user to an opened facility so that all opened facilities have at least $r$ users and the maximum distance from users to the assigned facilities is minimized~\cite{armon2011min}.

Both problems have plenty of practical applications.
One typical application is a privacy protection~\cite{sweeney2002k}:
Consider a company that publishes clustered data about their customers.
If there is a too small cluster, one can easily identify the individuals in a such cluster.
To guarantee the anonymity, we require the clusters to have at least $r$ individuals, which is obtained via $r$-gather clustering problem.
Another example is the formation of sport-teams.
Suppose that we want to divide people into several teams for a football game.
Then, each team must contain at least eleven people.
Such grouping can be obtained via $11$-gather clustering problem.

Several tractability and intractability results are known for both problems.
If $\mathcal{M}$ is a general metric space, there is a $3$-approximation algorithm for $r$-gathering problem, and no algorithm can achieve a better approximation ratio unless P=NP~\cite{armon2011min}.
If the set of locations of the users is a subset of that of the facilities\footnote{This version problem is originally called $r$-gather clustering problem~\cite{aggarwal2010achieving}.}, there is a $2$-approximation algorithm~\cite{aggarwal2010achieving}, and no algorithm can achieve a better approximation ratio unless P=NP~\cite{armon2011min}.
%
%These problems on restricted shaped metric are current main theme of investigation.
If $\mathcal{M}$ is a line, there are polynomial time exact algorithms by dynamic programming for $r$-gathering problem~\cite{akagi2015r,han2016r,nakano2018simple,sarker2019r}, where the fastest one runs in linear time~\cite{sarker2019r}.
This technique can also be applied to the $r$-gather clustering problem.
If $\mathcal{M}$ is a \emph{spider}, which is a metric space constructed by joining $d$ half-line-shaped metrics together at endpoints, these are polynomial-time algorithms if $d$ is a constant~\cite{ahmed2019r} (i.e., these are XP algorithms).
%however, these running times quite large: $O(r^d)$ for $r$-gather clustering problem and $O((2(r+1))^d)$ for $r$-gathering problem.
%i.e., these are XP algorithms in $d$.

Thus far, the best tractability result is the XP algorithms on a spider, and the best intractability is NP-hardness on a general metric space.
The purpose of this study is to close this gap.

\subsection{Our Contribution}

In this study, we close the gap between the tractability and intractability of the problems on a spider. % which is one of the simplest metrics other than a path.
We first propose fixed-parameter tractable algorithms for both $r$-gather clustering problem and $r$-gathering problem parameterized by $d$. 
%(i.e., $r$ can be a part of input).
\begin{theorem}
\label{tractable}
There is an algorithm to solve the $r$-gather clustering problem on a spider in $O(2^dr^4d^5 + n)$ time, where $d$ is the degree of the center of the spider.
Similarly, there is an algorithm to solve the $r$-gathering problem on a spider in $O(2^dr^4d^5 + n + m)$ time.
\end{theorem}
The proof is given in Section~3.

We then show that the problem is NP-hard; this is our main theoretical contribution.
\begin{theorem}\label{hardness}
The min-max $r$-gather clustering problem and min-max gathering problem are NP-hard even if the input is a spider.
\end{theorem}
The proof is given in Section~4.
Note that, on a spider, $r$-gather clustering problem is the special case of $r$-gathering problem because we can reduce the former into latter by putting facilities for all midpoints of two users.
Thus we only prove the NP-hardness only for $r$-gather clustering problem.

\begin{comment}
\subsection{Organization}

The rest of this paper is organized as follows. In Section 2 we give some basic definitions and notations. In Section 3 we give a FPT algorithm for min-max $r$-gather clustering problem and min-max $r$-gathering on a spider. In Section 4 we show that min-max $r$-gather clustering problem is NP-hard, even on a spider. %It automatically means $r$-gathering on spider is also NP-hard, since on spider, $r$-gather clustering is the special case of $r$-gathering.
\end{comment}

\section{Preliminaries}

A \emph{spider} $\mathcal{L}=\{l_1,\dots,l_d\}$ is a set of half-lines that share the endpoint $o$ (see Figure~(a)).
Each half-line is called a \emph{leg} and $o$ is called the \emph{center}.
We denote by $(l, x) \in \mathcal{L} \times \mathbb{R}_+$ the point on leg $l$ whose distance from the center is $x$.
Note that $(l, 0)$ is the center for all $l$.
The metric on $\mathcal{L}$ is defined by $d((l, x), (l', x')) = |x - x'|$ if $l = l'$ and $x + x'$ if $l \neq l'$.

Let $\mathcal{U}=\{u_1, \dots, u_n\}$ be a set of $n$ users on $\mathcal{L}$. 
A \emph{cluster} is a subset of users and the \emph{diameter} of a cluster is the distance between two farthest users in the cluster.

The \emph{min-max $r$-gather clustering on spider} is a problem to find a partition of users into disjoint clusters so that each cluster contains at least $r$ users (Figure~(b)).
The goal is to minimize the maximum diameter of the clusters (Figure~(c)).
When the meaning is clear, we simply call it \emph{$r$-gather clustering problem}.
The \emph{min-max $r$-gathering on spider} additionally specifies a set of $m$ facilities on $\mathcal{L}$.
The goal of this problem is to open a subset of facilities and assign each user to an opened facility so that each cluster that corresponds to a facility contains at least $r$ users (Figure~(d)).
The goal is to minimize the maximum distance from users to the assigned facilities (Figure~(e)).
When the meaning is clear, we simply call it {\it $r$-gathering}.

%In $r$-gathering, the optimal value doesn't change if we allow to assign a facility to multiple clusters. So, we do not care about it.

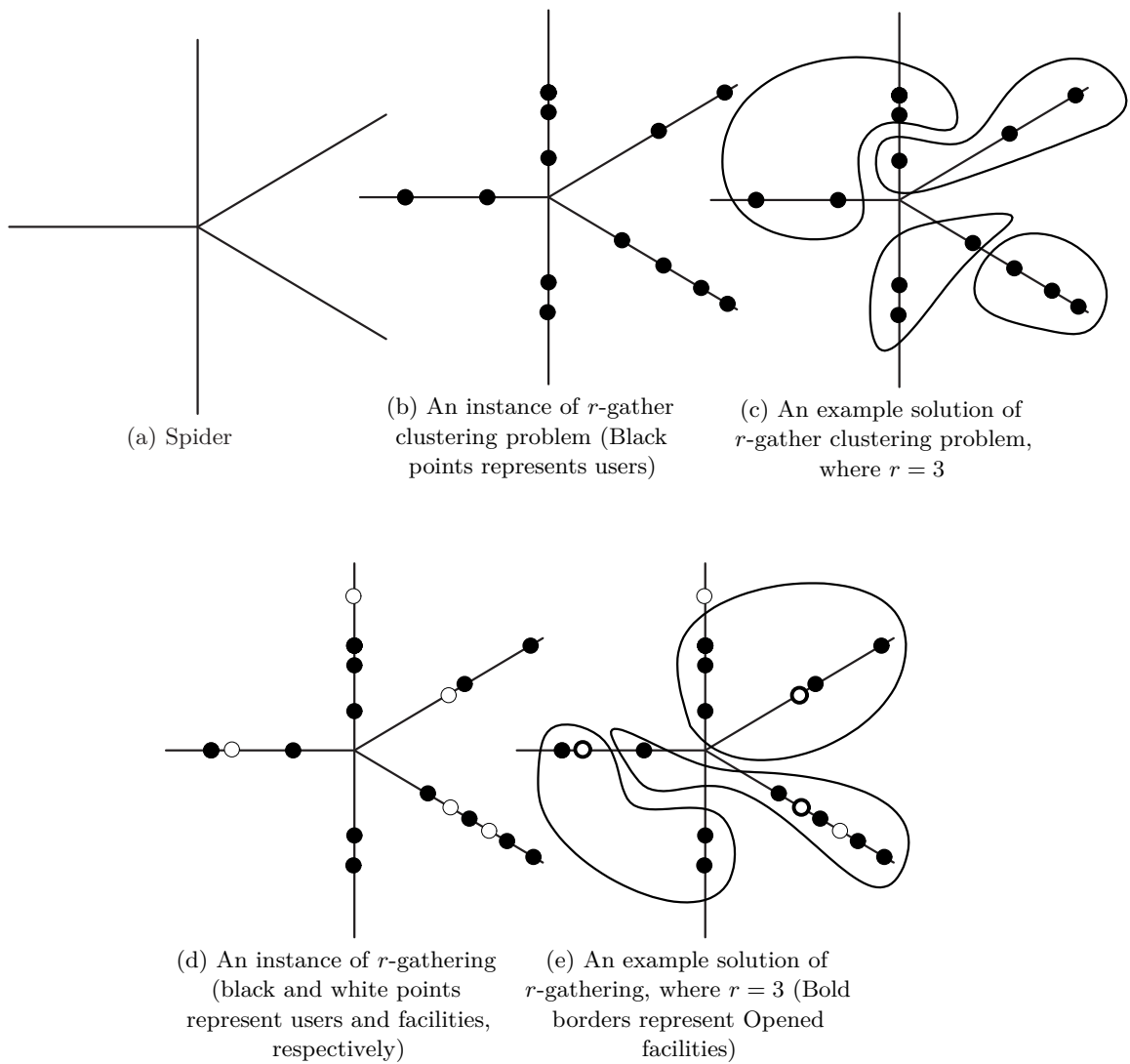
\begin{figure}[hbtp] 
\begin{center}
\begin{tabular}{c}
    \begin{minipage}{0.33\hsize}
        \begin{center}
          %WinTpicVersion4.32a
{\unitlength 0.1in%
\begin{picture}(20.0000,20.0000)(0.0000,-20.0000)%
% LINE 1 0 3 0 Black White  
% 10 0 1000 1000 1000 1000 0 1000 1000 1000 1000 2000 400 1000 1000 2000 1600 1000 1000 1000 2000
% 
\special{pn 13}%
\special{pa 0 1000}%
\special{pa 1000 1000}%
\special{fp}%
\special{pa 1000 0}%
\special{pa 1000 1000}%
\special{fp}%
\special{pa 1000 1000}%
\special{pa 2000 400}%
\special{fp}%
\special{pa 1000 1000}%
\special{pa 2000 1600}%
\special{fp}%
\special{pa 1000 1000}%
\special{pa 1000 2000}%
\special{fp}%
\end{picture}}%
          \hspace{1.6cm} (a) Spider
        \end{center}
    \end{minipage}
    \begin{minipage}{0.33\hsize}
        \begin{center}
          %WinTpicVersion4.32a
{\unitlength 0.1in%
\begin{picture}(20.0000,20.0000)(0.0000,-20.0000)%
% LINE 1 0 3 0 Black White  
% 10 0 1000 1000 1000 1000 0 1000 1000 1000 1000 2000 400 1000 1000 2000 1600 1000 1000 1000 2000
% 
\special{pn 13}%
\special{pa 0 1000}%
\special{pa 1000 1000}%
\special{fp}%
\special{pa 1000 0}%
\special{pa 1000 1000}%
\special{fp}%
\special{pa 1000 1000}%
\special{pa 2000 400}%
\special{fp}%
\special{pa 1000 1000}%
\special{pa 2000 1600}%
\special{fp}%
\special{pa 1000 1000}%
\special{pa 1000 2000}%
\special{fp}%
% CIRCLE 2 0 3 0 Black White  
% 4 240 1000 240 960 280 1000 280 1000
% 
\special{pn 8}%
\special{ar 240 1000 40 40 0.0000000 6.2831853}%
% CIRCLE 2 0 0 0 Black White  
% 4 240 1000 280 1000 280 1000 280 1000
% 
\special{sh 1.000}%
\special{ia 240 1000 40 40 0.0000000 6.2831853}%
\special{pn 8}%
\special{ar 240 1000 40 40 0.0000000 6.2831853}%
% CIRCLE 2 0 0 0 Black Black  
% 4 675 1000 675 960 715 1000 715 1000
% 
\special{sh 1.000}%
\special{ia 675 1000 40 40 0.0000000 6.2831853}%
\special{pn 8}%
\special{ar 675 1000 40 40 0.0000000 6.2831853}%
% CIRCLE 2 0 0 0 Black Black  
% 4 1000 440 1000 400 1040 440 1040 440
% 
\special{sh 1.000}%
\special{ia 1000 440 40 40 0.0000000 6.2831853}%
\special{pn 8}%
\special{ar 1000 440 40 40 0.0000000 6.2831853}%
% CIRCLE 2 0 3 0 Black White  
% 4 1000 440 1000 400 1040 440 1040 440
% 
\special{pn 8}%
\special{ar 1000 440 40 40 0.0000000 6.2831853}%
% CIRCLE 2 0 3 0 Black White  
% 4 1000 440 1000 400 1040 440 1040 440
% 
\special{pn 8}%
\special{ar 1000 440 40 40 0.0000000 6.2831853}%
% CIRCLE 2 0 3 0 Black White  
% 4 1000 440 1000 400 1040 440 1040 440
% 
\special{pn 8}%
\special{ar 1000 440 40 40 0.0000000 6.2831853}%
% CIRCLE 2 0 0 0 Black Black  
% 4 1950 1570 1950 1530 1990 1570 1990 1570
% 
\special{sh 1.000}%
\special{ia 1950 1570 40 40 0.0000000 6.2831853}%
\special{pn 8}%
\special{ar 1950 1570 40 40 0.0000000 6.2831853}%
% CIRCLE 2 0 0 0 Black Black  
% 4 1810 1485 1810 1445 1850 1485 1850 1485
% 
\special{sh 1.000}%
\special{ia 1810 1485 40 40 0.0000000 6.2831853}%
\special{pn 8}%
\special{ar 1810 1485 40 40 0.0000000 6.2831853}%
% CIRCLE 2 0 0 0 Black Black  
% 4 1610 1365 1610 1325 1650 1365 1650 1365
% 
\special{sh 1.000}%
\special{ia 1610 1365 40 40 0.0000000 6.2831853}%
\special{pn 8}%
\special{ar 1610 1365 40 40 0.0000000 6.2831853}%
% CIRCLE 2 0 0 0 Black Black  
% 4 1000 790 1000 750 1040 790 1040 790
% 
\special{sh 1.000}%
\special{ia 1000 790 40 40 0.0000000 6.2831853}%
\special{pn 8}%
\special{ar 1000 790 40 40 0.0000000 6.2831853}%
% CIRCLE 2 0 0 0 Black Black  
% 4 1000 545 1000 505 1040 545 1040 545
% 
\special{sh 1.000}%
\special{ia 1000 545 40 40 0.0000000 6.2831853}%
\special{pn 8}%
\special{ar 1000 545 40 40 0.0000000 6.2831853}%
% CIRCLE 2 0 0 0 Black Black  
% 4 1585 645 1585 605 1625 645 1625 645
% 
\special{sh 1.000}%
\special{ia 1585 645 40 40 0.0000000 6.2831853}%
\special{pn 8}%
\special{ar 1585 645 40 40 0.0000000 6.2831853}%
% CIRCLE 2 0 0 0 Black Black  
% 4 1935 440 1935 400 1975 440 1975 440
% 
\special{sh 1.000}%
\special{ia 1935 440 40 40 0.0000000 6.2831853}%
\special{pn 8}%
\special{ar 1935 440 40 40 0.0000000 6.2831853}%
% CIRCLE 2 0 0 0 Black Black  
% 4 1390 1230 1390 1190 1430 1230 1430 1230
% 
\special{sh 1.000}%
\special{ia 1390 1230 40 40 0.0000000 6.2831853}%
\special{pn 8}%
\special{ar 1390 1230 40 40 0.0000000 6.2831853}%
% CIRCLE 2 0 0 0 Black Black  
% 4 1000 1455 1000 1415 1040 1455 1040 1455
% 
\special{sh 1.000}%
\special{ia 1000 1455 40 40 0.0000000 6.2831853}%
\special{pn 8}%
\special{ar 1000 1455 40 40 0.0000000 6.2831853}%
% CIRCLE 2 0 0 0 Black Black  
% 4 995 1615 995 1575 1035 1615 1035 1615
% 
\special{sh 1.000}%
\special{ia 995 1615 40 40 0.0000000 6.2831853}%
\special{pn 8}%
\special{ar 995 1615 40 40 0.0000000 6.2831853}%
\end{picture}}%
          \hspace{1.6cm} (b) An instance of $r$-gather clustering problem (Black points represents users)
        \end{center}
    \end{minipage}
    \begin{minipage}{0.33\hsize}
        \begin{center}
          \input{clusterans.tex}
          \hspace{1.6cm} (c) An example solution of $r$-gather clustering problem, where $r=3$
        \end{center}
    \end{minipage}
\end{tabular}
\end{center}
\end{figure}

\begin{figure}[hbtp] 
\begin{center}
\begin{tabular}{c}
    \begin{minipage}{0.33\hsize}
        \begin{center}
          %WinTpicVersion4.32a
{\unitlength 0.1in%
\begin{picture}(20.0000,20.0000)(0.0000,-20.0000)%
% LINE 1 0 3 0 Black White  
% 10 0 1000 1000 1000 1000 0 1000 1000 1000 1000 2000 400 1000 1000 2000 1600 1000 1000 1000 2000
% 
\special{pn 13}%
\special{pa 0 1000}%
\special{pa 1000 1000}%
\special{fp}%
\special{pa 1000 0}%
\special{pa 1000 1000}%
\special{fp}%
\special{pa 1000 1000}%
\special{pa 2000 400}%
\special{fp}%
\special{pa 1000 1000}%
\special{pa 2000 1600}%
\special{fp}%
\special{pa 1000 1000}%
\special{pa 1000 2000}%
\special{fp}%
% CIRCLE 2 0 3 0 Black White  
% 4 240 1000 240 960 280 1000 280 1000
% 
\special{pn 8}%
\special{ar 240 1000 40 40 0.0000000 6.2831853}%
% CIRCLE 2 0 0 0 Black White  
% 4 240 1000 280 1000 280 1000 280 1000
% 
\special{sh 1.000}%
\special{ia 240 1000 40 40 0.0000000 6.2831853}%
\special{pn 8}%
\special{ar 240 1000 40 40 0.0000000 6.2831853}%
% CIRCLE 2 0 0 0 Black Black  
% 4 675 1000 675 960 715 1000 715 1000
% 
\special{sh 1.000}%
\special{ia 675 1000 40 40 0.0000000 6.2831853}%
\special{pn 8}%
\special{ar 675 1000 40 40 0.0000000 6.2831853}%
% CIRCLE 2 0 0 0 Black Black  
% 4 1000 440 1000 400 1040 440 1040 440
% 
\special{sh 1.000}%
\special{ia 1000 440 40 40 0.0000000 6.2831853}%
\special{pn 8}%
\special{ar 1000 440 40 40 0.0000000 6.2831853}%
% CIRCLE 2 0 3 0 Black White  
% 4 1000 440 1000 400 1040 440 1040 440
% 
\special{pn 8}%
\special{ar 1000 440 40 40 0.0000000 6.2831853}%
% CIRCLE 2 0 3 0 Black White  
% 4 1000 440 1000 400 1040 440 1040 440
% 
\special{pn 8}%
\special{ar 1000 440 40 40 0.0000000 6.2831853}%
% CIRCLE 2 0 3 0 Black White  
% 4 1000 440 1000 400 1040 440 1040 440
% 
\special{pn 8}%
\special{ar 1000 440 40 40 0.0000000 6.2831853}%
% CIRCLE 2 0 0 0 Black Black  
% 4 1950 1570 1950 1530 1990 1570 1990 1570
% 
\special{sh 1.000}%
\special{ia 1950 1570 40 40 0.0000000 6.2831853}%
\special{pn 8}%
\special{ar 1950 1570 40 40 0.0000000 6.2831853}%
% CIRCLE 2 0 0 0 Black Black  
% 4 1810 1485 1810 1445 1850 1485 1850 1485
% 
\special{sh 1.000}%
\special{ia 1810 1485 40 40 0.0000000 6.2831853}%
\special{pn 8}%
\special{ar 1810 1485 40 40 0.0000000 6.2831853}%
% CIRCLE 2 0 0 0 Black Black  
% 4 1610 1365 1610 1325 1650 1365 1650 1365
% 
\special{sh 1.000}%
\special{ia 1610 1365 40 40 0.0000000 6.2831853}%
\special{pn 8}%
\special{ar 1610 1365 40 40 0.0000000 6.2831853}%
% CIRCLE 2 0 0 0 Black Black  
% 4 1000 790 1000 750 1040 790 1040 790
% 
\special{sh 1.000}%
\special{ia 1000 790 40 40 0.0000000 6.2831853}%
\special{pn 8}%
\special{ar 1000 790 40 40 0.0000000 6.2831853}%
% CIRCLE 2 0 0 0 Black Black  
% 4 1000 545 1000 505 1040 545 1040 545
% 
\special{sh 1.000}%
\special{ia 1000 545 40 40 0.0000000 6.2831853}%
\special{pn 8}%
\special{ar 1000 545 40 40 0.0000000 6.2831853}%
% CIRCLE 2 0 0 0 Black Black  
% 4 1585 645 1585 605 1625 645 1625 645
% 
\special{sh 1.000}%
\special{ia 1585 645 40 40 0.0000000 6.2831853}%
\special{pn 8}%
\special{ar 1585 645 40 40 0.0000000 6.2831853}%
% CIRCLE 2 0 0 0 Black Black  
% 4 1935 440 1935 400 1975 440 1975 440
% 
\special{sh 1.000}%
\special{ia 1935 440 40 40 0.0000000 6.2831853}%
\special{pn 8}%
\special{ar 1935 440 40 40 0.0000000 6.2831853}%
% CIRCLE 2 0 0 0 Black Black  
% 4 1390 1230 1390 1190 1430 1230 1430 1230
% 
\special{sh 1.000}%
\special{ia 1390 1230 40 40 0.0000000 6.2831853}%
\special{pn 8}%
\special{ar 1390 1230 40 40 0.0000000 6.2831853}%
% CIRCLE 2 0 0 0 Black Black  
% 4 1000 1455 1000 1415 1040 1455 1040 1455
% 
\special{sh 1.000}%
\special{ia 1000 1455 40 40 0.0000000 6.2831853}%
\special{pn 8}%
\special{ar 1000 1455 40 40 0.0000000 6.2831853}%
% CIRCLE 2 0 0 0 Black Black  
% 4 995 1615 995 1575 1035 1615 1035 1615
% 
\special{sh 1.000}%
\special{ia 995 1615 40 40 0.0000000 6.2831853}%
\special{pn 8}%
\special{ar 995 1615 40 40 0.0000000 6.2831853}%
% CIRCLE 3 0 2 0 Black White  
% 4 995 175 995 135 1035 175 1035 175
% 
\special{sh 0}%
\special{ia 995 175 40 40 0.0000000 6.2831853}%
\special{pn 4}%
\special{ar 995 175 40 40 0.0000000 6.2831853}%
% CIRCLE 3 0 2 0 Black White  
% 4 1510 1305 1510 1265 1550 1305 1550 1305
% 
\special{sh 0}%
\special{ia 1510 1305 40 40 0.0000000 6.2831853}%
\special{pn 4}%
\special{ar 1510 1305 40 40 0.0000000 6.2831853}%
% CIRCLE 3 0 2 0 Black White  
% 4 1500 705 1500 665 1540 705 1540 705
% 
\special{sh 0}%
\special{ia 1500 705 40 40 0.0000000 6.2831853}%
\special{pn 4}%
\special{ar 1500 705 40 40 0.0000000 6.2831853}%
% CIRCLE 3 0 2 0 Black White  
% 4 1715 1430 1715 1390 1755 1430 1755 1430
% 
\special{sh 0}%
\special{ia 1715 1430 40 40 0.0000000 6.2831853}%
\special{pn 4}%
\special{ar 1715 1430 40 40 0.0000000 6.2831853}%
% CIRCLE 3 0 2 0 Black White  
% 4 350 995 350 955 390 995 390 995
% 
\special{sh 0}%
\special{ia 350 995 40 40 0.0000000 6.2831853}%
\special{pn 4}%
\special{ar 350 995 40 40 0.0000000 6.2831853}%
\end{picture}}%
          \hspace{1.6cm} (d) An instance of $r$-gathering (black and white points represent users and facilities, respectively)
        \end{center}
    \end{minipage}
    \begin{minipage}{0.33\hsize}
        \begin{center}
          \input{r-gatheringans.tex}
          \hspace{1.6cm} (e) An example solution of $r$-gathering, where $r=3$ (Bold borders represent Opened facilities)
        \end{center}
    \end{minipage}
\end{tabular}
\end{center}
\end{figure}

\section{FPT Algorithm for $r$-Gather Clustering and $r$-Gathering on Spider}

We give an FPT algorithm to solve min-max $r$-gather clustering problem and $r$-gathering problem on spider parametrized by the number of legs $d$.
First, we exploit the structure of optimal solutions.
After that, we give a brute-force algorithm. Finally, we accelerate it by dynamic programming.

We denote the coordinate of user $u$ by $(l(u),x(u))$.
Without loss of generality, we assume that $x(u_1) \le \dots \le x(u_n)$. 
%that $x(u_i)\leq x(u_j)$ holds for all $1\leq i<j\leq n$.
We refer this order to explain a set of users: for example, ``the first/last $k$ users on leg $l$'' means the users with $k$ smallest/largest index among all users on leg $l$. Users on the center are located on every leg, but we choose an arbitrary leg and consider that all the users are located on this leg. Thus every user is considered to be located on exactly one leg.

%\begin{dfn}
We introduce a basic lemma about the structure of a solution.
A cluster is \emph{single-leg} if it contains users from a single leg; otherwise, it is \emph{multi-leg}.
Ahmed et al.~\cite{ahmed2019r} showed that there is an optimal solution that has a specific single-leg/multi-leg structure as follows.
%A \emph{single-leg cluster} is a cluster that only contains users from single legs. 
%Otherwise, a cluster is {\it multi-leg cluster}.
%\end{dfn}
%This is already presented in previous paper~\cite{ahmed2019r}, so we omit the proof.

\begin{lemma}[{\rm \cite[Lemma~2]{ahmed2019r}}]
For both $r$-gather clustering problem and $r$-gathering problem, there is an optimal solution such that, for all leg $l$, first some users on $l$ are contained in multi-leg clusters and rest of them are contained in single-leg clusters.
\end{lemma}

For a while, we concentrate on the structure of multi-leg clusters.
%\begin{dfn}
Let $C$ be a multi-leg cluster.
Let $u_i$ be the last user in $C$ and $u_j$ be the last user with $l(u_i)\neq l(u_j)$ in $C$.
A \emph{ball part} of $C$ is the set of users whose indices are at most $u_j$ and a \emph{segment part} of $C$ is the set of the remaining users. %set of the rest of the users.
$C$ is \emph{special} if $C$ contains all the users on $l(u_i)$ and the ball part is $\{u_1, \dots, u_k\}$ for some integer $k$.
The list of multi-leg clusters $\{C_1, \dots ,C_t\}$ are {\it suffix-special} if for all $1\leq i\leq t$, $C_i$ is a special when we only consider the users in $C_i, \dots ,C_t$.
%\end{dfn}

The following lemma is the key to our algorithm.
Since it is a reformulation of Lemma~3 and Lemma~8 in Ahmed et al.~\cite{ahmed2019r} using Lemma~2 in Nakano~\cite{nakano2018simple}, we omit the proof.
%, but we do not need the whole part of the statement.
%Here we reformulate the lemma 3 and lemma 8 in it, combined with the lemma which was proposed by Nakano~\cite{nakano2018simple}.

\begin{lemma}[{\rm Reformulation of \cite[Lemmas~3 and 8]{ahmed2019r} by \cite[Lemma~2]{nakano2018simple}}]
%The following is satisfied for both of $r$-gather clustering and $r$-gathering: 
Suppose that there exists an optimal solution without any single-leg cluster.
Then there is an optimal solution such that all the clusters contain at most $2r-1$ users, and there exists a special cluster.
\end{lemma}

By definition, the segment part of a cluster is non-empty and contains the users from a single leg. 
By removing a special cluster and applying the lemma repeatedly, we can state that there is an optimal solution consists of a suffix-special family of multi-leg clusters.

\begin{algorithm}[tb]
\caption{Suffix-Special Multi-Leg Clusters Generation\label{subalg}}    
\begin{algorithmic}[1]
\REQUIRE Set of legs $\mathcal{L}=\{l_1, \dots, l_d\}$, set of users $\mathcal{U}=\{u_1, \dots, u_n\}$, positive integer $r$
\STATE $\mathcal{C}:=\emptyset, C:=\emptyset, S:=\{1, \dots, d\}$
\FOR{$i=1, \dots, n$}
    \IF{$l(u_i) \in S$}
        \STATE Choose one. (a): Use $u_i$ in ball part of cluster $C$. $C:=C\cup\{u_i\}$.\\ (b): Discard $u_i$ and all further users in $l(u_i)$. $S:=S\setminus \{l(u_i)\}$. In this case, discarded users are used in single-leg clusters.
    \ENDIF
    \STATE Choose one. (c): Continue to choose ball part. Do nothing. We can choose this alternative only when $|C|<2r-1$. \\ (d): Finish to choose the ball part and go on to choose the segment part.
    \IF{(d) is chosen}
        \STATE Choose a leg $l\in S$.
        \STATE Choose a non-negative integer $t$, subject to there are at least $t$ users on leg $l$ whose indices are larger than $i$ and $r\leq |C|+t\leq 2r-1$. Add first $t$ users among such users to $C$.
        \STATE $S:=S\setminus \{l\}$, $\mathcal{C}:=\mathcal{C}\cup \{C\}, C:=\emptyset$
    \ENDIF
\ENDFOR
\ENSURE $\mathcal{C}$ (only when $C$ is empty)
\end{algorithmic}
\end{algorithm}

We first present a brute-force algorithm that enumerates all the suffix-special families of multi-leg clusters in Algorithm~\ref{subalg}. 
The correctness is clear from the definition.
Once all such families of clusters are enumerated, the remaining thing is to deal with single-leg clusters.
It can be done by solving the line-case problem on the remaining users, independently for each leg, which can be solved by pre-calculated dynamic programming.
%We can easily see that this algorithm generates all special lists of multi-leg clusters.

%Now we give a brute-force algorithm to enumerate all the special family of multi-leg clusters.
%We show the algorithm in Algorithm \ref{subalg}. 

Then, we accelerate Algorithm~\ref{subalg} by dynamic programming. Similar to Algorithm~\ref{subalg}, we make a special cluster of remaining users one by one, by looking through the users and decide whether to use them into the current cluster.
In Algorithm~\ref{subalg}, all information we should remember to construct the current cluster is $i$, $S$, and $C$.
We claim that, for the dynamic programming, we only need to remember (1) the size of $C$ in order not to make a too-small cluster and (2) the index of the last user in the ball part of $C$ in order to calculate the diameter/cost of the cluster. 
Assume that we know the last user $u$ in the ball part of $C$ and last user $v$ in the segment part of $C$.
Then, the diameter/cost of $C$ is spanned by $u$ and $v$; 
Then, we can compute the diameter/cost of the cluster.
%It is because these pair of two users forms diameter in $r$-gather clustering and these two users are the candidates of the users farthest from the facility in $r$-gathering.
Here we denote the diameter/cost of the multi-leg cluster by $\Cost(v,u)$ for both of the problems.

We should also deal with single-leg clusters.
For single-leg clusters, we can apply the results for line case, which are well studied.
For all leg $l$ and for all integers $k$ from $0$ to the number of users on leg $l$, we first compute the optimal objective value, only considering the last $k$ users on leg $l$.
For each user $u_i$, we denote the optimal objective value for the set of users on leg $l(u_i)$ whose indices are greater than $i$ and no less than $i$ by $R^+(u_i)$ and $R^-(u_i)$, respectively.
All these values can be computed in linear time for both $r$-gather clustering problem and $r$-gathering problem, by using the technique in the latest method~\cite{sarker2019r}.

\begin{algorithm}[tb]
\caption{A FPT-time algorithm of $r$-gather clustering problem and $r$-gathering problem on spider}         
\label{alg}             
\begin{algorithmic}[1]
\REQUIRE Set of legs $\mathcal{L}=\{l_1, \dots, l_d\}$, set of users $\mathcal{U}=\{u_1, \dots, u_n\}$, positive integer $r$. In $r$-gathering, we are also given a set of facilities $\mathcal{F}$.
\STATE Calculate $R^+(u_i),R^-(u_i)$ for all $i$.
\STATE $\DP[i][S][j][k]:= \infty$ for all $0\leq i\leq n, S\subseteq \{1,\dots d\}, 0\leq j\leq 2r, 0\leq k\leq 2dr$
\COMMENT{Here $i,S,j,k$ means the user we are looking at now, the set of available legs, the size of current cluster, the last user in current cluster, respectively}
\STATE $\DP[0][S][0][0]:=0$
\FOR{$i=1,\dots,n$}
    \FOR{$S\subseteq \{1,\dots,d\}, j=0,\dots,2r-2, k=0,\dots,i-1$, such that $l(u_i)\in S$}
        \STATE $\DP[i][S][j+1][i]:=\min(\DP[i][S][j+1][i],\DP[i-1][S][j][k])$
        \COMMENT{Here we use $u_i$ in the ball part of current cluster}
        \STATE $\DP[i][S\setminus \{l(u_{i})\}][j][k]:=\min(\DP[i][S\setminus \{l(u_{i})\}][j][k],\max(\DP[i-1][S][j][k],R^-(u_i)))$
        \COMMENT{Here we discard leg $l(u_i)$}
    \ENDFOR
    \FOR{$S\subseteq \{1,\dots,d\}, j=0,\dots,2r-2, k=1,\dots,i-1, l=1,\dots,d$, such that $l\not \in S$}
         \FOR{$p=\max(0,r-j),\dots,2r-1-j$}
             \IF{There are at least $p$ users on leg $l$ whose indices are larger than $i$}
             \STATE Let $v$ be the $p$-th such user
             \STATE $\DP[i][S\setminus \{l\}][0][i]=\min(\DP[i][S\setminus\{l\}][0][i],$ $\max(\DP[i][S][j][k], \Cost(v,u_k), R^+(v)))$
             \COMMENT{Here we use remaining first $p$ users on leg $l$ as the segment part of the current cluster}
             \ENDIF
         \ENDFOR
    \ENDFOR
\ENDFOR
\STATE $res:=\infty$
\FOR{$i=0,\dots,n$}
    \STATE $res=\min(res,\DP[i][\emptyset][0][i])$
\ENDFOR
\ENSURE $res$
\end{algorithmic}
\end{algorithm}

The algorithm is shown in Algorithm~\ref{alg}.
The correctness is clear from the construction.
Thus, we analyze the time complexity.
A naive implementation of the algorithm requires $O(2^dn^2r^2d)$ evaluations of $\Cost$ and a preprocessing for $R^+$ and $R^-$.
Each evaluation of $\Cost$ requires $O(1)$ time for $r$-gather clustering problem and $O(m)$ time for $r$-gathering problem.
The preprocessing requires $O(n)$ time for $r$-gather clustering problem and $O(n+m)$ time for $r$-gathering problem~\cite{sarker2019r}.
Thus the time complexities are $O(2^dn^2r^2d)$ for $r$-gather clustering problem and $O(2^dn^2r^2dm)$ for $r$-gathering problem.

We can further improve the complexity.
The loop for $i$ is reduced to see only first $(2r-1)d$ users from each leg since other users cannot be contained in the ball part of multi-leg clusters.
Thus we can replace $n$ to $rd^2$ in the complexity so as we obtain the complexities $O(2^dr^4d^5+n)$ for $r$-gather clustering problem and $O(2^dr^4d^5m+n)$ for $r$-gathering problem.
In $r$-gather clustering problem, this is a linear-time algorithm when $d,r$ are sufficiently smaller than $n$.

In $r$-gathering problem, we can further improve the complexity by improving the algorithm to calculate $\Cost$ (see Appendix~A).
When $d,r$ are sufficiently smaller than $n,m$, this is also a linear-time algorithm.
Therefore, Theorem~\ref{tractable} is proved.

\section{NP-hardness of min-max $r$-Gather Clustering on Spider}

We prove the min-max $r$-gather clustering problem is NP-hard even on a spider.
We first propose \emph{arrears problem} (Problem~\ref{arrears}) as an intermediate problem.
Then, we reduce the arrears problem to the min-max $r$-gather clustering problem on a spider.
Finally, we prove the strong NP-hardness of the arrears problem.

\begin{prb}[Arrears Problem]
\label{arrears}
We are given $n$ sets $S_1, \dots, S_n$ of pairs of integers, i.e., $S_i = \{ (a_{i,1}, p_{i,1}), \dots, (a_{i,|S_i|}, p_{i,|S_i|}) \}$ for all $i = 1, \dots, n$.
We are also given $m$ pairs of integers $(b_1, q_1), \dots, (b_m, q_m)$.
%$j$-th of them is represented as a pair of integer $(b_j,q_j)$.
The task is to decide whether there is $n$ integers $z_1, \dots, z_n$ such that
%decide whether we can choose exactly one index $z_i$ for each $i=1,\dots,n$, which satisfies 
%We are given $n$ {\it payment duties} of the sets $S_1,\dots,S_n$. Each payment duty $S_i$ consists of pairs of integers $(a_{i,1},p_{i,1}),\\ \dots,(a_{i,|S_i|},p_{i,|S_i|})$.
%We are also given $m$ {\it budget constraints}, $j$-th of them is represented as a pair of integer $(b_j,q_j)$. Decide whether we can choose exactly one index $z_i$ for each $i=1,\dots,n$, which satisfies 
\begin{align*}
%\label{arrear-ineq}
\sum_{a_{i,z_i}\leq b_j}p_{i,z_i}\leq q_j
\end{align*}
holds for for all $j=1, \dots, m$. 
\end{prb}
The name of the ``arrears problem'' comes from the following interpretation.
Imagine a person who is in arrears in his $n$ \emph{payment duties} $S_1, \dots, S_n$.
Each payment duty $S_i$ has multiple options $(a_{i,1}, p_{i,1}), \dots, (a_{i, |S_i|}, p_{i, |S_i|})$ such that he can choose a \emph{payment amount} of $p_{i, k}$ dollar with the \emph{payment date} $a_{i, k}$ for some $k$.
Each pair $(b_j, q_j)$ corresponds to his \emph{budget constraint} such that he can pay at most $q_j$ dollar until $b_j$-th day.
%Intuitively, this problem represent the following situation. The person is in arrears his/her $n$ payment duties. He/She should pay all the payment duties in some day. $i$-th payment duty should be paid at {\it payment date} $a_{i,1}$ or,\dots,or $a_{i,|S_i|}$, but late payment costs more. Payment on $a_{i,k}$-th day costs $p_{i,k}$ amount of money, called {\it payment amount}. Additionally, he/she cannot pay no more than $q_j$ amount of money until $b_j$-th day, to save life.

The arrears problem itself may be an interesting problem, but here we use this problem just for a milestone to prove the hardness of min-max $r$-gather clustering problem on a spider.
The proof follows the following two propositions.

\begin{prp}[Reduction from the arrears problem]\label{hard1}
If the arrears problem is strongly NP-hard, the min-max $r$-gather clustering problem on a spider is NP-hard.
\end{prp}
\begin{prp}[Hardness of the arrears problem]\label{hard2}
The arrears problem is strongly NP-hard.
\end{prp}

Without loss of generality, we assume that $b_1 < \dots < b_m$ and $q_1 < \dots < q_m$. 
Also, we assume that $a_{i,1} < \dots < a_{i,|S_i|}$ and $p_{i,1} < \dots < p_{i,|S_i|}$ for all $i = 1, \dots, n$.

\subsection{Reduction from Arrears Problem}

We first prove Proposition~\ref{hard1}.
In this subsection, let $n$ be the number of payment duties and $m$ be the number of budget constraints.

Let $\mathcal{I}$ be an instance of the arrears problem.
We define $L = \max(\max_{i}a_{i,|S_i|},b_m)+1$ and 
$r = \max(\max_{i}p_{i,|S_i|},q_m)+1$.
We construct an instance $\mathcal{I}'$ of the decision version of the $r$-gather clustering problem on a spider that requires to decide whether there is a way to divide vertices into clusters all of which has size at least $r$ and diameter at most $2L$.

%which is equivalent to given instance and whose number of vertices (so do the number of legs) is polynomial of $n,m$ and input values. 

In construction, we distinguish two types of legs --- long and short.
Each long leg corresponds to a payment duty and each short leg corresponds to a budget constraint.
%We simply call this leg {\it long leg $i$}, or simply {\it leg $i$}.

%For each payment duty $S_i$, we define a long leg $i$ of length $4L-a_{i,|S_i|}+1$.
For each payment duty $S_i$, we define a long leg $i$.
We first put $r$ users on $(i, 4L-a_{i,|S_i|}+1)$. Then, we put $p_{i,k+1}-p_{i,k}$ users on $(i,2L-a_{i,k})$ for all $k=1,\dots,|S_i|-1$. 
Finally, we put $r-p_{i,|S_i|}$ users on $(i,2L-a_{i,|S_i|})$.

A short leg has only one user. The distance from the center to the user is referred to as the \emph{length} of the short leg.
For each $j = 1, \dots, m$, we define $q_j - q_{j-1}$ short legs of length $b_{j-1} + 1$, where we set $q_0 = b_0 = 0$.
We also define $r$ short legs of length $L$.

%Short legs are used to represent budget constraint. Each short leg has only one vertex, hence we represent a short leg just by its length. For each $j=1,\dots,m$, we make $q_{i}-q_{i-1}$ short legs with length $b_{j-1}+1$, where we set $q_0=b_0=0$ for convenience. Finally, we prepare $r$ short legs of length $L$.

This construction can be done in pseudo-polynomial time.
Now we prove that $\mathcal{I}'$ has a feasible solution if and only if $\mathcal{I}$ is a yes-instance of the arrears problem.
We start by looking some basic structures of clusters in a feasible solution of $\mathcal{I}'$.

\begin{lemma}\label{twolegs}
In a feasible solution to $\mathcal{I}'$, there is no cluster that contains users from two different long legs.
\end{lemma}
\begin{proof}
By definition, the distance between the center and a user on a long leg is larger than $L$.
So, the distance between users from two different long legs exceeds $2L$; hence, they cannot be in the same cluster.
\end{proof}

An \emph{end cluster} of long leg $i$ is a cluster that contains the farthest user of $i$.
%on endpoint of this leg.
%We call such a cluster {\it end cluster} of leg $i$. 
The above lemma implies that, in a feasible solution, end clusters of different long legs are different. 
Intuitively, the ``border'' of end cluster of long leg $i$ corresponds to the choice from the options of payment duty $S_i$.

\begin{lemma}\label{threesmall}
For each long leg $i$, the following three statements hold.
{\rm (a)} An end cluster of $i$ only contains the users from $i$.
{\rm (b)} There is exactly one end cluster of $i$, and no other cluster consists of only users from leg $i$.
{\rm (c)} Some users on $i$ are not contained in tend cluster.
\end{lemma}

\begin{proof}
{\rm (a)} The endpoint of $i$ is distant by more than $2L$ from center.
{\rm (b)} There are less than $2r$ users on $i$ so they cannot form more then one clusters alone.
{\rm (c)} Users on the point $(i,2L-a_{i,|S_i|})$ are distant from the endpoint of $i$ by more than $2L$, thus they cannot be in the same cluster.
\end{proof}

Lemma~\ref{twolegs} and the third statement of Lemma~\ref{threesmall} implies that users who are not contained in end clusters should form a cluster together with users from short legs.
Now, we prove Proposition~\ref{hard1}.

%Now, preparation to prove proposition \ref{hard1} is completed. Following lemma means the proposition \ref{hard1}.
%\begin{lem}
%Assume that a feasible solution of $\mathcal{I}'$ is given. Then, there is a polynomial time algorithm to construct a feasible solution of $\mathcal{I}$. Conversely, given a feasible solution of $\mathcal{I}$, we can construct a feasible solution of $\mathcal{I}'$.
%\end{lem}

\begin{proof}[Proof of Proposition~\ref{hard1}]
%Assume that a feasible solution of $\mathcal{I}'$ is given. Then, there is a polynomial time algorithm to construct a feasible solution of $\mathcal{I}$. Conversely, given a feasible solution of $\mathcal{I}$, we can construct a feasible solution of $\mathcal{I}'$.

Suppose that we have a feasible solution to the instance of the min-max $r$-gathering problem on a spider that is constructed as the above.
For each long leg $i$, let $u_i$ be the last user that is not contained in end clusters, and $C_i$ be the cluster that contains $u_i$.
Then, the location of $u_i$ is represented as $(i,2L-a_{i,z_i})$ by using an integer $z_i$.
We choose payment date $a_{i,z_i}$ for payment duty $i$.
We prove that these choices of payment dates are a feasible solution to the arrears problem.

As described above, $C_i$ consists of users from leg $i$ and short legs. Since on leg $i$ there are only $(r-p_{i,|S_i|})+(p_{i,|S_i|}-p_{i,|S_i|-1})+ \dots +(p_{i,z_i+1}-p_{i,z_i})=r-p_{i,z_i}$ users on the path from the center to the location of $u_i$, $C_i$ should contain at least $p_{i,z_i}$ users on short legs with length at most $a_{i,z_i}$.
For $j$-th budget constraint, by the rule of construction, there are $(q_1-q_0)+ \dots +(q_j-q_{j-1})=q_j$ users on short legs whose length is at most $b_j$.
Suppose $a_{i,z_i}\leq b_j$. We use at least $p_{i,z_i}$ users on short legs whose lengths are at most $a_{i,z_i}\leq b_j$ in the cluster $C_i$. Thus, the sum of $p_{i,z_i}$ among all $i$ with $a_{i,z_i}\leq b_j$ is at most the number of users on short legs whose length is at most $b_j$, that is, $q_j$.
That means the budget constraint holds.

Conversely, suppose that we are given a feasible solution to the instance $\mathcal{I}$ of the arrears problem.
First, for each payment duty $i$ we make a cluster with all users located between $(i,4L-a_{i,|S_i|}+1)$ and $(i, 2L-a_{i,z_i}+1)$, inclusive.
This cluster contains at least $r$ users since there are $r$ users on point $(i,4L-a_{i,|S_i|}+1)$ and has diameter at most $2L$.
We renumber the payment duties (thus so do long legs) in the non-decreasing order of $a_{i,z_i}$ and proceed them through the order of indices: 
For a payment duty $i = 1, 2, \dots, n$, we make a cluster $C_i$ by all remaining users on leg $i$ and all users from remaining $p_{i,z_i}$ shortest short legs.
By the construction, these clusters have exactly $r$ users.
We show that the diameter of $C_i$ is at most $2L$.
The diameter is spanned by a long leg and the longest short leg. 
The distance to the long leg in $C_i$ is $2 L - a_{i, z_i}$.
The longest short leg in $C_i$ is the $p_{1,z_1}+ \dots +p_{i,z_i}$-th shortest short leg.
We take the smallest $j$ such that $a_{i,z_i}\leq b_j$.
Then, since given solution is a feasible solution to $\mathcal{I}$, $p_{1,z_1}+ \dots +p_{i,z_i}\leq q_j$ holds.
Since there are $q_j$ users on short legs with length less than $b_{j-1}+1\leq a_{i,z_i}$, the length of longest short leg in $C_i$ is at most $a_{i,z_i}$.
This gives the diameter of $C_i$ is at most $2 L$.
Finally, we make a cluster with all remaining users.
Since there are $r$ short legs of length $L$ and all these users are located within the distance $L$ from the center, we can just put them into a cluster.
Then we obtain a feasible solution to $\mathcal{I}'$, which completes the proof.
\end{proof}

\subsection{Strong NP-Hardness of Arrears Problem}

Now we give a proof outline of Proposition~\ref{hard2}; the full proof is given in Appendix~B.
We reduce the 1-IN-3SAT problem, which is known to be NP-complete~\cite{schaefer1978complexity}.
%We now give an outline of the reduction from 1-IN-3SAT problem.

\begin{prb}[\rm 1-IN-3 SAT problem~\cite{schaefer1978complexity}]
We are given a set of clauses, all of them contains exactly three literals. Decide whether there exists a truth assignment such that all clause has exactly one true literal.
\end{prb}
%In this subsection we denote the number of variables and clauses by $n,m$, respectively. We can assume that $n,m\geq 1$.

%We do not give a formal construction here and write an outline instead. We give a formal construction and hardness proof in appendix B. 

\begin{proof}[Proof Outline of Proposition~\ref{hard2}]
Let $n$ and $m$ be the number of boolean variables and clauses, respectively.
For each variable $x_i$, we prepare $N := 3 m (m+2) + 1$ items $T_i$ for the positive literal $x_i$ and $N$ items $\bar{T}_i$ for negative literal $\bar{x}_i$.
Let $T = \bigcup_i (T_i \cup \bar{T}_i)$ be the set of all items.
%Let's start constructing $\mathcal{I}'$.
%Now we start by defining the set of payment duties.
%To make out discussion clear, we prepare an item for each payment duty we construct.
%We prepare $2(3m(m+2)+1):=2N$ items (thus payment duties) for each variable $x_i$.
%$N$ of them correspond to literal $x_i$ and remaining $N$ items correspond to literal $\bar{x}_i$.
%We denote the set of items corresponding to $x_i$ by $T_i$ and $\bar{x}_i$ by $\bar{T}_i$.
Each item $y \in T$ corresponds to a payment duty $\{ (a_{y,1}, p_{y,1}), (a_{y,2}, p_{y,2}) \}$ having two options.
Then, a solution to the arrears problem is specified by a set $X \subseteq T$ of items $y$ such that $a_{y, 2}$ is chosen.
We denote by $\bar{X} = T \setminus X$ the complement of $X$.
%Each payment duty has exactly two components.
%We denote the payment duty corresponding to item $y$ by $\{(a_{y,1},p_{y,1}),(a_{y,2},p_{y,2})\}$.
%That means, $X$ is the set of items corresponds to the literals which is set to be $\texttt{true}$.
%Now we describe the roles of both parts of days and introduce further budget constraints to ensure these roles. 
We want to construct a solution to the 1-IN-3SAT problem from a solution $X$ to the arrears problem by 
%$x_i = \texttt{false}$ if $a_{y,1}$ is chosen for $y \in T_i$; otherwise, $x_i = \texttt{false}$.
$x_i = \texttt{true}$ if $y \in X$ for some $y \in T_i$; otherwise $x_i = \texttt{false}$.
To make this construction valid, we define payment dates and payment amounts suitably as follows.

The payment days consist of two periods: the first period is $\{1, \dots, n\}$ and the second period is $\{n+1, \dots, n+m+2\}$.
For each item $y$, $a_{y,1}$ belongs to the first period and $a_{y,2}$ belongs to the second period. 
%The payment amount $p_{y,1}$ is given in the form $B^4+\alpha_yB^3+  a_{y,1} B^2+a_{y,1} \alpha_y B+ (a_{y,2}-n-1)$ 
Let $a_{y,1}=i$ and $a_{y,2}=n+1+j$.
Then, the payment amount $p_{y,1}$ is given in the form $B^4+\alpha_yB^3+iB^2+i\alpha_yB+j$
%:=e_{y,4}B^4+e_{y,3}B^3+e_{y,2}B^2+e_{y,1}B+e_{y,0}$, 
where $B$ is a sufficiently large integer, and $\alpha_y$ is a non-negative integer, where $\sum_{y\in T_i}\alpha_y=\sum_{y\in \bar{T}_i}\alpha_y=N$ holds for all $i$.
%In the following construction, we sum up values of $p_{y,i}$. 
%Since $B$ is sufficiently large, we do not need to care about the carries. % when we sum up these values. %$e_{y,4},\dots,e_{y,0}$ are uniquely defined to satisfy $0\leq e_{y,4},\dots,e_{y,0}\leq B-1$.
We define $p_{y,2} = 2p_{y,1}$ for all $y \in T$.
%Thus, we can identify value $B^4 + e_{y,3} B^3 + e_{y,2} B^2 + e_{y,1} B^1 + e_{y,0}$ as a tuple $(e_{y,3}, e_{y,2}, e_{y,1}, e_{y,0})$.

Let $R = (1/2) \sum_{y \in T} p_{y,1} = n N B^4 + n N B^3 + n(n+1)/2 N B^2 + n(n+1)/2 N B + \dots$. 
We make two budget constraints $(n, R)$ and $(n+m+2, 3R)$.
Then, these constraints hold in equality: Let $x \le R$ be the total payment until $n$. Then the total payment until $n + m + 2$ is $x + 2 (2 R - x) = 4 R - x \le 3 R$. These inequalities imply $x = R$. (see Lemma~\ref{halfeq} on Appendix B).
%Since $p_{y,2} = 2p_{y,1}$, if the sum of the chosen payment amounts till day $n$ increases, the sum of the chosen payment amounts till $n+m+2$ decreases by the same value.
%These budget constraints are tight in this meaning: if the equation does not hold in budget constraints on day $n$, the sum of the chosen payment amounts till day $n+m+2$ exceeds the budget constraint $3R$.
%Thus the only way to fulfill both inequalities is to fulfill both inequalities in equations.

We use the first period to ensure that the truth assignment produced by $X$ is well-defined, i.e., if $y \in X$ for some $y \in T_i$ then $y' \in X$ for all $y' \in T_i$.
%a feasible solution of $\mathcal{I}'$ is well-defined. % i.e., if payment date $a_{y,1}$ (resp. $a_{y,2}$) is chosen for some $y\in T_i$, then $a_{y',1}$ (resp. $a_{y',2}$) should be chosen for every $y' \in T_i$.
%and $a_{y,1}$ should be chosen for every $y\in \bar{T}_i$ (and vice versa).
% To ensure this, we add a budget constraint $(i, q_i)$ such that $q_i^{(4)} = q_i^{(3)} = i N$.
First, for each $i=1, \dots, n$, we add a budget constraint $(i,iNB^4+iNB^3+(B^3-1))$. By comparing the coefficients of $B^4$ and $B^3$, we have
\[
\sum_{y\in \bar{X}\cap \bigcup_{j=1}^i(T_j\cup \bar{T}_j)} (B^4 + a_y B^3) \leq i N B^4 + i N B^3.
\]
We can prove that for all $i$ these inequalities hold in equality, i.e., 
\begin{align}
\label{welldefeq}
\sum_{y \in y \in \bar{X}\cap (T_i\cup \bar{T}_i)} (B^4 + a_y B^3) = N B^4 + N B
\end{align}
holds for all $i$ as follows.
%e_{y,3}$ of $$ is $NB+N$:
By using the relation between the coefficients of $p_{y,1}$, we have $\sum_{y \in \bar{X}} (i B^2 + i a_y B) \ge \frac{n(n+1)}{2} N B^2 + \frac{n(n+1)}{2} N B$ (see Proposition~\ref{weldef} on Appendix B).
%To ensure this, we bound the sum of $e_{y,2}B + e_{y,1}$ in $\bar{X}$.
%Since $e_{y,2}B + e_{y,1}=i(B + e_{y,3})$ for all $i$, the sum of $e_{y,2}B + e_{y,1}$ in $\bar{X}$ is lower-bounded by $\frac{n(n+1)}{2}NB+\frac{n(n+1)}{2}N$ (see Proposition~\ref{weldef} on Appendix B). 
Since the budget constraint $(n,R)$ is fulfilled in equality, and the coefficients of $B^2$ and $B$ in $R$ are both $\frac{n(n+1)}{2} N$, this inequality holds in equality, which implies \eqref{welldefeq}.
%the coefficients of it should be also equal. Since the coefficients of the $B^2$ and $B$ are both $\frac{n(n+1)}{2}$, which is equal to the lower bound, every inequalities we used here should be fulfilled in equation.
%Therefore, we have $\sum_{y\in \bar{X}\cap(T_i\cap \bar{T}_i)} B^4 + a_y B^3 = NB+N$.
Then, we define values $\alpha_y$ appropriately so that the only $X \cap (T_i\cap \bar{T}_i) = T_i$ or $X \cap (T_i\cap \bar{T}_i)=\bar{T}_i$ satisfy equation \eqref{welldefeq} (see Proposition~\ref{weldef} on Appendix B).
This ensures the well-definedness of the truth assignment.

The second period represents the clauses. Let $Z_i$ be the set of items with $a_{y,2}=i$.
We put budget constraint $(i,(nN + 2\sum_{j=n+1}^i K_j)B^4+(B^4-1))$ for each $i=n+1, \dots, n+m+2$, where $K_{n+1}, \dots, K_{n+m+2}$ are non-negative integers determined later.
Then, by a similar argument to the first period, we can prove that 
\[
|\bar{X}|+2|X\cap (Z_{n+1}\cup\dots\cup Z_{i})|=nN+2\sum_{j=n+1}^iK_j
\]
for each $i=n+1,\dots,n+m+2$ (see Proposition~\ref{clause} in Appendix B).
This implies that $|X\cap Z_i|=K_i$ for each $i = n+1, \dots, n+m+2$.
%
%Now we define the values $K_{n+1},\dots,K_{n+m+2}$.
The budget constraint on day $i \ge n + 3$ corresponds to the $i-(n+2)$-th clause.
For $i=n+3, \dots, n+m+2$, we set $K_i=1$.
Then, we have $|X\cap Z_i|=1$, i.e., exactly one literal in $i-(n+2)$-th clause is $\mathtt{true}$. 
The budget constraints on day $n+1$ and $n+2$ are used for the adjustment.
%, which are determined by $K_{n+1}=|X\cap Z_{n+1}|$ and $K_{n+2}=|X\cap Z_{n+2}|$ from the feasibility of 1-IN-3SAT: 
Since $\{ Z_{n+1}, \dots, Z_{n+m+2} \}$ forms a partition of items, we have $|X \cap Z_{n+1}|+|X \cap Z_{n+2}|=|X|-(|X\cap Z_{n+3}| + \dots + |X\cap Z_{n+m+2}|) = N-m$.
Moreover, since the constant term $e_{y,0}$ of $p_{y,1}$ is $e_{y,0} = i-(n+1)$ for all $y\in Z_i$ and $i=n+1,\dots,n+m+2$, we have $\sum_{y\in X} e_{y,0} = \sum_{i=n+1}^{n+m+2}(i-(n+1))|X\cap Z_i|$.
By solving these equations, we obtain $K_{n+1} = |X \cap Z_{n+1}|$ and $K_{n+1} = |X \cap Z_{n+2}|$. 
%we obtain these 
%and the left term is equal to the constant coefficient of $R$, the value $|X\cap Z_{n+2}|$ can be specified.
%Thus, we can also determine the value $|X\cap Z_{n+1}|$ and finished to construct a gadget.
Since all value appears in $\mathcal{I}'$ is at most $2B^4$ and we can take $B$ in a polynomial of $n,m$.
Thus, The hardness proof is completed.
\end{proof}

\bibliographystyle{plainurl}% the mandatory bibstyle
\bibliography{bib.bib}

\newpage

\appendix

\section{Calculation of $\Cost$ in $r$-gathering}

In this section, we show how to calculate the $\Cost(v,u)$ in the $r$-gathering problem efficiently. The number of candidates of pair $v,u$ is at most $O(r^2d^4)$, so we calculate the $\Cost$ for all candidates in advance and store them. Now, we describe how to calculate these values. We assume that the facilities are given in increasing order of the distances from the center.

There are two cases of the location of the facility which will be assigned to the cluster -- whether the facility is located on the leg $l(v)$ or not. If it is not located on the leg $l(v)$, we should simply choose the facility which is nearest to the center. This case can be processed in constant time for each pair of $v,u$.

For the facility located on $l(v)$, we should choose the facility which is nearest to the midpoint of the coordinates of $v$ and $u$. We enumerate the midpoints for all pairs and sort them by distance from the center, for each leg. We can apply the two-pointer technique to find the optimal facility by seeing the midpoints in sorted order. This case can be processed in $O(r^2d^4\log(rd)+m)$ time, where $\log$ came from the sorting operation.

We can refer each pre-calculated $\Cost$ value in $O(1)$ time, so the total time complexity of Algorithm \ref{alg} is reduced to $O(2^dr^4d^5+r^2d^4\log(rd)+n+m)=O(2^dr^4d^5+n+m)$.

\section{Proof of NP-Hardness of Arrears Problem}

In this section, we construct an instance of arrears problem $\mathcal{I}'$ from given instance of 1-IN-3SAT $\mathcal{I}$.
In our construction, every payment duty has exactly two payment dates. Let us fix the variable $x_i$. For all $j=1,\dots,m$ and $k=1,2,3$, we prepare two items $u_{i,j,k}$ and $\bar{u}_{i,j,k}$. We also prepare auxiliary items $w_{i,l}$ and $\bar{w}_{i,l}$ for each $l=1,\dots,3m(m+1)+1$. In this way, We prepare $6m(m+2)+2$ items in total for each $x_i$.

We name some important sets of items in following way:

\begin{itemize}
    \item $U_i=\{u_{i,j,k}|1\leq j\leq m,1\leq k\leq 3\}$
    \item $\bar{U}_i=\{\bar{u}_{i,j,k}|1\leq j\leq m,1\leq k\leq 3\}$
    \item $W_i=\{w_{i,l}|1\leq l\leq 3m(m+1)+1\}$
    \item $\bar{W}_i=\{\bar{w}_{i,l}|1\leq l\leq 3m(m+1)+1\}$
    \item $T_i=U_i\cup W_i$
    \item $\bar{T}_i=\bar{U}_i\cup \bar{W}_i$.
    \item $Y_i=T_i\cup \bar{T}_i$
\end{itemize}

We prepare a payment duty $S_y=\{(a_{y,1},p_{y,1}),(a_{y,2},p_{y,2})\}$ for all item $y\in Y_i$.

Before defining the value of these values, we take an integer $B$. All integers which appear as $p_{y,1}$ is represented in the form $e_{y,4}B^4+e_{y,3}B^3+e_{y,2}B^2+e_{y,1}B+e_{y,0}$ by non-negative integers $e_{y,4},\dots,e_{y,0}$. Similarly, all integers which appear as $q_j$ is represented in the form $f_{j,4}B^4+f_{j,3}B^3+f_{j,2}B^2+f_{j,1}B+f_{j,0}$ by non-negative integers $f_{j,4},\dots,f_{j,0}$. We represent these values as if like a vector $(e_{y,4},e_{y,3},e_{y,2},e_{y,1},e_{y,0})$ and $(f_{j,4},f_{j,3},f_{j,2},f_{j,1},f_{j,0})$.
In our construction $p_{y,2}$ is always equal to $2p_{y,1}$, so $p_{y,2}$ can be represented in the form $(2e_{y,4},2e_{y,3},2e_{y,2},2e_{y,1},2e_{y,0})$. 

We take $B$ to be sufficiently large value (but still in a polynomial of $n,m$) so that for all $k=4,3,2,1,0$, sum of $3e_{y,k}$ over all item $y$ is still less than $B$. 
That means that we can compare the sum of payment amount and budget constraint just by comparing sum of $(e_{y,4},e_{y,3},e_{y,2},e_{y,1},e_{y,0})$ (or $(2e_{y,4},2e_{y,3},2e_{y,2},2e_{y,1},2e_{y,0})$) and $(f_{j,4},f_{j,3},f_{j,2},f_{j,1},f_{j,0})$ by lexicographical order of five-dimensional vector.
The concrete value of $B$ is $100n^2m^2$.

Let us start to set the value of payment duties. Let us fix a variable $x_i$. For $u\in U_i\cup \bar{U}_i$, we set $a_{u_1}=i$ and 
\[
a_{u,2}=\left\{\begin{array}{ll}
n+2+j & (u=u_{i,j,k}\:\text{and}\:c_{j,k}=x_i)\\
n+2+j & (u=\bar{u}_{i,j,k}\:\text{and}\:c_{j,k}=\bar{x}_i)\\
n+1 & (\text{otherwise}).
\end{array}
\right.
\]
$p_{u,2}$ is always equal to $2p_{u,1}$. Remaining task is define the value $p_{u,1}$. We set $p_{u,1}$ by following formula. 
\[
\left\{\begin{array}{ll}
(B^2+i)(B+1)B+(j+1) & (u=u_{i,j,k}\:\text{and}\:c_{j,k}=x_i)\\
(B^2+i)(B+1)B & (u=u_{i,j,k}\:\text{and}\:c_{j,k}\neq x_i)\\
(B^2+i)B^2+(j+1) & (u=\bar{u}_{i,j,k}\:\text{and}\:c_{j,k}=\bar{x}_i)\\
(B^2+i)B^2 & (u=\bar{u}_{i,j,k}\:\text{and}\:c_{j,k}\neq \bar{x}_i)
\end{array}
\right.
\]
For each $w\in W_i\cup \bar{W}_i$,  let $K_i$ be $\sum_{u\in U_i}(p_{i,1} \mod B)$ and $\bar{K}_i$ be $\sum_{u\in \bar{U}_i}(p_{i,1} \mod B)$. We set $a_{w,1}=i$ and 
\[
a_{w,2}=\left\{\begin{array}{ll}
n+2 & (w=w_{i,l}\:\text{and}\:1\leq l\leq 3m(m+1)-K_i)\\
n+1 & (w=w_{i,l}\:\text{and}\:3m(m+1)-K_i+1\leq l\leq 3m(m+1)+1)\\
n+2 & (w=\bar{w}_{i,l}\:\text{and}\:1\leq l\leq 3m(m+1)-\bar{K}_i)\\
n+1 & (w=\bar{w}_{i,l}\:\text{and}\:3m(m+1)-\bar{K}_i+1\leq l\leq 3m(m+1)+1).
\end{array}
\right.
\]
$p_{w,2}$ is also always equal to $2p_{w,1}$. Remaining task is define the value $p_{w,1}$. We set $p_{w,1}$ by following formula. 
\[
\left\{\begin{array}{ll}
(B^2+i)(B+1)B+1 & (w=w_{i,l}\:\text{and}\:1\leq l\leq 3m(m+1)-K_i)\\
(B^2+i)(B+1)B & (w=w_{i,l}\:\text{and}\:3m(m+1)-K_i+1\leq l\leq 3m(m+1))\\
(B^2+i)(B+1)B & (w=w_{i,l}\:\text{and}\:l=3m(m+1)+1)\\
(B^2+i)B^2+1 & (w=\bar{w}_{i,l}\:\text{and}\:1\leq l\leq 3m(m+1)-\bar{K}_i)\\
(B^2+i)B^2 & (w=\bar{w}_{i,l}\:\text{and}\:3m(m+1)-\bar{K}_i+1\leq l\leq 3m(m+1))\\
(B^2+i)(B+3m(m+2)+1)B & (w=\bar{w}_{i,l}\:\text{and}s\:l=3m(m+1)+1).
\end{array}
\right.
\]
Note that, since $K_i,\bar{K}_i\leq 3m(m+1)$ by definition, 
\begin{eqnarray}
\sum_{t\in T_i}p_{t,1}=\sum_{t\in \bar{T}_i}p_{t,1}=(3m(m+2)+1)(B^2+i)(B+1)B+3m(m+1)\nonumber
\end{eqnarray}
holds for all $i$. We call this value $R_i$ and set
\[
R=\sum_{i=1}^nR_i=(3m(m+2)+1)(B+1)(nB^2+\frac{n(n+1)}{2})B+3m(m+1)n.
\]
We now remark that $B$ is sufficiently large. We can calculate each $e_{y,k}$ values just by expanding the definition formula. It can be calculated that, for all $k=4,3,2,1,0$ the sum of $e_{y,k}$ over all $p_{y,1}$ is at most $2(3m(m+2)+1)(\frac{n(n+1)}{2})\leq 20m^2n^2$. It is sufficiently small to avoid carry.
%\begin{comment}
%When we simply expand the definition of $p_{t,1},p_{t,2}$ as if definition formula is polynomial of $B$, we can represent these values in the form of $x_0B^4+x_1B^3+x_2B^2+x_3B+x_4$. Since sum of $x_0,x_1,x_2,x_3,x_4$ is at most $6(3m(m+2)+1)(\frac{n(n+1)}{2})\leq 60m^2n^2<B$, we can ignore the influence of carry. That is, when we are concerning about the sum of coefficient $x_i$, we can just ignore the lower digit.
%\end{comment}
Then, we set budget constraints. We set a budget constraints $(i,q_i)$ for each $i=1,\dots,n+m+2$. Value of $q_i$ is
\[
\left\{\begin{array}{ll}
(3m(m+2)+1)(B+1)iB^3+(B^3-1) & (1\leq i\leq n-1)\\
R & (i=n)\\
((3m(m+2)+1)n+6mn+2n+m(m+1))B^4+(B^4-1) & (i=n+1)\\
(3(3m(m+2)+1)n-2(n+m+2-i))B^4+(B^4-1) & (n+2\leq i\leq n+m+1)\\
3R & (i=n+m+2).
\end{array}\right.
\]
We represent $q_i=\{f_{i,4},f_{i,3},f_{i,2},f_{i,1},f_{i,0}\}$.

Now we complete our construction.
All appearing values are at most $3R=O(nm^2B^4)=O(n^9m^{10})$, which is bounded in a polynomial of $n,m$.
We prove that $\mathcal{I}$ is a yes-instance of 1-IN-3SAT if and only if $\mathcal{I}'$ is a yes-instance of arrears problem.

For feasible solution of $\mathcal{I}'$, Let $X$ be the set of items $t$, such that payment date $p_{t,2}$ is chosen for payment duty $S_t$. We define $\bar{X}$ by complement of $X$. Intuitively, for $y=u_{i,j,k}$ or $y=w_{i,l}$, $y\in X$ means $x_i$ is true and $y \in \bar{X}$ means $x_i$ is false. Following proposition guarantees that this type of truth assignment is well-defined.

\begin{prp}\label{weldef}
In a feasible solution of $\mathcal{I}'$, for all $1\leq i\leq n$, one of the following condition holds.
\begin{itemize}
    \item $X\cap Y_i=T_i$.
    \item $X\cap Y_i=\bar{T}_i$.
\end{itemize}
\end{prp}

Before proving this proposition, we prove the following basic property.

\begin{lemma}\label{halfeq}
In a feasible solution of $\mathcal{I}'$, 
\[
\sum_{y\in \bar{X}}p_{y,1}=\sum_{y\in X}p_{y,1}=R.
\]
\end{lemma}
\begin{proof}
We only needs budget constraints for $n$ and $n+m+2$ to prove this lemma. Note that,
\[
\sum_{y\in \bar{X}}p_{y,1}+\sum_{y\in X}p_{y,1}=\sum_{i=1}^n\left(\sum_{y\in T_i}p_{y,1}+\sum_{y\in \bar{T}_i}p_{y,1}\right)=2R
\]
holds. From budget constraint for $n$,
\[
\sum_{y\in \bar{X}}p_{y,1}\leq R
\]
holds. From budget constraint for $n+m+2$,
\[
\sum_{y\in \bar{X}}p_{y,1}=4R-\left(\sum_{y\in \bar{X}}p_{y,1}+2\sum_{y\in X}p_{y,1}\right)\geq 4R-3R=R
\]
holds. It means equality holds for both of the inequality, and thus the lemma holds.
\end{proof}

\begin{proof}

Let us start by rephrasing some budget constraints.
First, we only concern about the coefficients of $B^4$ and $B^3$. From budget constraint for $i=1,\dots,n$,
\[
\sum_{y\in \bar{X}\cap (Y_1\cup\dots\cup Y_i)}(e_{y,4},e_{y,3}) \leq (3m(m+2)+1)i(1,1)
\]
holds for all $i=1,\dots,n-1$. Since $i(e_{y,4},e_{y,3})=(e_{y,2},e_{y,1})$ holds for all $y\in Y_i$, 

\begin{eqnarray}
(f_{n,2},f_{n,1})&=&\sum_{y\in \bar{X}}(e_{y,2},e_{y,1})\nonumber \\
&=&\sum_{i=1}^{n}\sum_{y\in \bar{X}\cap Y_i}i(e_{y,4},e_{y,3})\nonumber \\
&=&\sum_{i=1}^{n}\sum_{y\in \bar{X}\cap Y_i}n(e_{y,4},e_{y,3})-\sum_{i=1}^{n-1}\sum_{y\in \bar{X}\cap (Y_1\cup\dots\cup Y_i)}(e_{y,4},e_{y,3})\nonumber \\
&\geq&n\sum_{y\in \bar{X}}(e_{y,4},e_{y,3})-\sum_{i=1}^{n-1}(3m(m+2)+1)i(1,1)\nonumber \\
&=&((3m(m+2)+1)(n^2-\frac{(n-1)n}{2})(1,1)\nonumber \\
&=&(f_{n,2},f_{n,1})\nonumber
\end{eqnarray}
holds. So, 
\[
\sum_{y\in \bar{X}\cap (Y_1\cup\dots\cup Y_i)}(e_{y,4},e_{y,3}) = (3m(m+2)+1)i(1,1)
\]
holds for all $i=1,\dots,n-1$. For $i=n$ this equality also holds because of lemma \ref{halfeq}. That means, 
\[
\sum_{y\in \bar{X}\cap Y_i}(e_{y,4},e_{y,3}) = (3m(m+2)+1)(1,1)
\]
holds for all $i=1,\dots,n$. We can see that only in the situation $\bar{X}\cap Y_i=T_i$ and $\bar{X}\cap Y_i=\bar{T}_i$ this equation holds by definition of $e_{y,4}$ and $e_{y,3}$. Since $X$ is complement of $\bar{X}$, that completes the proof.
\end{proof}

Now we consider clauses. We define the set of items $Z_j$ for all $j=0,\dots,m+1$ as the set of items $y$ with $e_{y,0}=j$. Following proposition ensures that exactly one variable in each clause is true. We prove this proposition by a similar technique of the proof of proposition \ref{weldef}.

\begin{prp}\label{clause}
In a feasible solution of $\mathcal{I}'$, for all $j=2,\dots,m+1$, $|X\cap Z_j|=1$ holds.
\end{prp}
\begin{proof}
First we concern about only the coefficient of $B^4$. From the budget constraints for $n+1,\dots,n+m+1$ and lemma \ref{halfeq},
\[
\sum_{y\in X\cap Z_0}2e_{y,4}\leq 6mn+2n+m(m+1)
\]
and for all $j=1,\dots,m$,
\[
\sum_{y\in X\cap (Z_0\cup\dots\cup Z_j)}2e_{y,4}\leq 2(3m(m+2)+1)n-2(m+1-j)
\]
holds. Since $e_{y,0}=je_{y,4}$ holds for all $y\in Z_j$,

\begin{eqnarray}
f_{n+m+2,0}-f_{n,0}&=&\sum_{y\in X}2e_{y,0}\nonumber \\
&=&\sum_{j=0}^{m+1}\sum_{y\in X\cap Z_j}2je_{y,4}\nonumber \\
&=&\sum_{j=0}^{m+1}\sum_{y\in X\cap Z_j}2(m+1)e_{y,4}-\sum_{j=0}^{m}\sum_{y\in X\cap (Z_0\cup\dots\cup Z_j)}2e_{y,4}\nonumber \\
&\geq&2(m+1)\sum_{j=0}^{m+1}\sum_{y\in X\cap Z_j}e_{y,4}-(6mn+2n+m(m+1))\nonumber \\
&&-\sum_{j=1}^{m}(2(3m(m+2)+1)n-2(m+1-j))\nonumber \\
&=&2(3m(m+2)+1)(m+1)n-(6mn+2n+m(m+1))\nonumber \\
&&-2(3m(m+2)+1)nm+2m(m+1)-m(m+1)\nonumber \\
&=&6mn(m+1)=f_{n+m+2,0}-f_{n,0}\nonumber
\end{eqnarray}
holds. So, for all $j=1,\dots,m$,
\[
\sum_{y\in X\cap (Z_0\cup\dots\cup Z_j)}2e_{y,4} = 2(3m(m+2)+1)n-2(m+1-j)
\]
holds. That means $|X\cap Z_j|= 1$ holds for all $j=2,\dots,m+1$.

\end{proof}

Now we understood the solution structure for $\mathcal{I}'$ enough to prove the hardness result. The following proposition forms the reduction.

\begin{prp}
Assume that feasible solution of $\mathcal{I}'$ is given. Then, there is a polynomial time algorithm to construct a feasible solution of $\mathcal{I}$. Conversely, given a feasible solution of $\mathcal{I}$, we can construct a feasible solution of $\mathcal{I}'$.
\end{prp}

\begin{proof}
Assume that we are given a feasible solution of $\mathcal{I}'$. For all $i=1,\dots,n$, we set $x_i$ to be true when $T_i\cup X=T_i$ and false otherwise. By the rule of construction and Propositions~\ref{weldef}, \ref{clause}, it is a feasible solution of $\mathcal{I}$.

Conversely, assume that we are given a feasible solution of $\mathcal{I}$. For all $i=1,\dots,n$, we choose $T_i$ if $x_i$ is true and $\bar{T}_i$ otherwise. We set $X$ to be the union of all chosen sets of items.

All we have to show is this solution satisfies all budget constraint. Since
\[
\sum_{y\in \bar{X}\cap (Y_1\cup\dots\cup Y_i)}(e_{y,4},e_{y,3})=(3m(m+2)+1)i(1,1)
\]
budget constraints for $i=1,\dots,n-1$ holds. Since
\[
\sum_{y\in \bar{X}}p_{y,1}=R
\]
budget constraints for $i=n$ and $i=n+m+1$ holds. By the rule of construction $|X\cap Z_j|=1$ for all $j=2,\dots,m+1$. So,
\[
\sum_{y\in \bar{X}}e_{y,4}+\sum_{y\in X\cap (Z_0\cup\dots\cup Z_j)}2e_{y,4}=3(3m(m+2)+1)n-2(m+1-j)
\]
holds and budget constraints for $i=n+2,\dots,n+m+1$ holds. Finally,
\begin{eqnarray}
|X\cap Z_0|&=&(3m(m+2)+1)n-|X\cap Z_1|-\sum_{j=2}^{m+1}|X\cap Z_j|\nonumber \\
&=&(3m(m+2)+1)n-3m(m+1)n+\sum_{j=2}^{m+1}\sum_{y\in Z_j}(e_{y,0}-1)\nonumber \\
&=&3mn+n+\frac{m(m+1)}{2}\nonumber
\end{eqnarray}
holds, then 
\[
\sum_{y\in \bar{X}}e_{y,4}+\sum_{y\in X\cap Z_0}2e_{y,4}\leq (3m(m+2)+1)n+6mn+2n+m(m+1)
\]
holds. So budget constraint for $i=n+1$ holds. That completes the proof.
\end{proof}

Above proposition completes the hardness proof. 

\end{document}